\documentclass[onecolumn, 11pt, letterpaper]{article}

\usepackage[margin = 1in]{geometry}

\usepackage{graphicx}
\usepackage{amsthm}
\usepackage{amsmath,amssymb,amsfonts,bm,setspace}\allowdisplaybreaks
\usepackage[utf8]{inputenc}
\usepackage{tikz}
\usepackage{algorithm,algorithmicx}
\usepackage[noend]{algpseudocode}
\usepackage{fancyhdr}
\usepackage{autobreak}
\usepackage{mathtools}
\usepackage{mathrsfs}
\usepackage{enumerate}
\usepackage{xifthen}
\usetikzlibrary{decorations.pathreplacing,decorations.markings,decorations.pathmorphing,decorations.shapes,arrows.meta,positioning}

\usepackage[hidelinks]{hyperref}

\DeclareMathOperator*{\argmax}{argmax}

\DeclareMathOperator*{\E}{E} 
\newcommand{\supp}{\mathrm{supp}} 

\newcommand{\compilehidecomments}{false}

\ifthenelse{ \equal{\compilehidecomments}{true} }{%
	\newcommand{\xiaoming}[1]{}
	\newcommand{\zhijie}[1]{}
	\newcommand{\jialin}[1]{}
	\newcommand{\shuo}[1]{}
}{
	\newcommand{\xiaoming}[1]{{\color{blue!50!black}  [\text{Xiaoming:} #1]}}
	\newcommand{\zhijie}[1]{{\color{red!60!black} [\text{Zhijie:} #1]}}
	\newcommand{\jialin}[1]{{\color{brown!60!black} [\text{Jialin:} #1]}}
	\newcommand{\shuo}[1]{{\color{green!60!black} [\text{Shuo:} #1]}}
}

\newtheorem{lemma}{Lemma}
\newtheorem{theorem}{Theorem}
\newtheorem{definition}{Definition}

\newcommand*\samethanks[1][\value{footnote}]{\footnotemark[#1]}

\begin{document}
\title{Improved Deterministic Algorithms for Non-monotone Submodular Maximization}
%
%

\author{
	Xiaoming Sun\thanks{Institute of Computing Technology, Chinese Academy of Sciences, Beijing, China, and School of Computer Science and Technology, University of Chinese Academy of Sciences, Beijing, China.
E-mail:   
\{sunxiaoming, zhangjialin, zhangshuo19z\}@ict.ac.cn}
\and
Jialin Zhang\samethanks
\and
Shuo Zhang\samethanks
\and
Zhijie Zhang\thanks{Center for Applied Mathematics of Fujian Province, School of Mathematics and Statistics, Fuzhou University, Fuzhou, China.
E-mail: zzhang@fzu.edu.cn}
}


%
%
%
\maketitle              

\begin{abstract}
Submodular maximization is one of the central topics in combinatorial optimization. It has found numerous applications in the real world. In the past decades, a series of algorithms have been proposed for this problem. However, most of the state-of-the-art algorithms are randomized. There remain non-negligible gaps with respect to approximation ratios between deterministic and randomized algorithms in submodular maximization. 

In this paper, we propose deterministic algorithms with improved approximation ratios for non-monotone submodular maximization. Specifically, for the matroid constraint, we provide a deterministic $0.283-o(1)$ approximation algorithm, while the previous best deterministic algorithm only achieves a $1/4$ approximation ratio. For the knapsack constraint, we provide a deterministic $1/4$ approximation algorithm, while the previous best deterministic algorithm only achieves a $1/6$ approximation ratio. For the linear packing constraints with large widths, we provide a deterministic $1/6-\epsilon$ approximation algorithm. To the best of our knowledge, there is currently no deterministic approximation algorithm for the constraints.

\end{abstract}
%
%
%
\section{Introduction}

Submodular maximization refers to the problem of maximizing a \emph{submodular} set function under some specific constraint.
The submodularity of the objective functions captures the effect of \emph{diminishing returns} in the economy and therefore the problem has found numerous applications in the real world, including influence maximization \cite{KKT03}, sensor placement and feature selection \cite{IB12,IB13}, information gathering \cite{KG11} and machine learning \cite{DK08,KED+17}.

Due to its remarkable significance, submodular maximization has been studied over the past forty years.
The difficulty of the problem is different depending on whether the objective function is \emph{monotone}.

For the monotone case, it is well-known that the problem can not be approximated within a ratio better than $1-1/e$ even under the cardinality constraint \cite{NemhauserW78,Fei98}.
On the other hand, a natural greedy algorithm achieves a $1-1/e$ approximation ratio for the cardinality constraint \cite{NemhauserWF78} and $1/2$ ratio for the matroid constraint \cite{FisherNW78}.
The greedy algorithm may have an arbitrarily bad approximation ratio for the knapsack constraint, but by combining the enumeration technique, it can be augmented to be a $1-1/e$ approximation algorithm \cite{SAJ99,Sviridenko04}.
It has remained a longstanding open question whether the problem admits a $1-1/e$ approximation ratio for the matroid constraint.
In 2008, Vondr{\'{a}}k \cite{Vondrak08} answered this question affirmatively by proposing the famous \emph{continuous greedy} algorithm.
Later, the technique was further generalized and applied to other constraints such as the linear packing constraints \cite{KulikST09} and general constraints with the down-closed property \cite{FeldmanNS11}.
Unlike the greedy algorithm, the continuous greedy algorithm requires a sampling process and therefore is inherently a randomized algorithm.

For the non-monotone case, the problem can not be approximated within a ratio better than $0.491$ under the cardinality constraint and $0.478$ under the matroid constraint \cite{GharanV11}.
The best known algorithm for the aforementioned constraints achieves a $0.385$ approximation ratio \cite{BuchbinderF19}.
The algorithm is randomized since it applies the continuous greedy technique.

When randomness is not allowed, existing algorithms in the literature for the non-monotone case suffer from inferior approximation ratios.
Currently, the best deterministic algorithm for the cardinality constraint achieves a $1/e$ approximation ratio \cite{BuchbinderF18}. The best one for the matroid constraint has an approximation ratio of $1/4$ \cite{HCCW20}.
The best one for the knapsack constraint has an approximation ratio of $1/6$ \cite{GuptaRST10}.
For linear packing constraints, there are currently no deterministic algorithms with constant ratios to the best of our knowledge.

The gap between deterministic and randomized algorithms inspires us to design better deterministic algorithms for the problem.
This is of great interest for two reasons.
From a theoretical viewpoint, it is an interesting and important question whether randomness is essentially necessary for submodular maximization.
From a practical viewpoint, randomized algorithms only work well in the average case, while deterministic algorithms still work in the worst case.
In this sense, deterministic algorithms are more robust and hence more suitable in some fields such as the security domain.

\subsection{Our Contribution}

In this paper, we provide several improved deterministic algorithms for non-monotone submodular maximization subject to different constraints. 
\begin{itemize}
	\item For the matroid constraint, we present a deterministic $0.283-o(1)$ approximation algorithm with running time $\mathcal{O}(n^2 k^2)$.
	\item For the knapsack constraint, we present a deterministic $1/4$ approximation algorithm and a slightly faster deterministic $1/4-\epsilon$ approximation algorithm. 
	Our knapsack algorithm borrows the idea of the Simultaneous Greedy algorithm in \cite{AmanatidisKS19} and the Twin Greedy algorithm in \cite{HCCW20}.
	We also present a tight example showing that this approach can not reach an approximation ratio better than $1/4$ even under a cardinality constraint.
	Hence, our analysis for the algorithm, as well as the analysis in \cite{HCCW20}, is tight.
	\item For the linear packing constraints, we present a deterministic $1/6-\epsilon$ approximation algorithm when the width $W=\min \{b_i/A_{ij}\mid A_{ij}>0\}$ satisfies $W=\Omega(\ln m/\epsilon^2)$.
	Our algorithm can be adapted to solve the instances with arbitrary widths when the number of constraints $m$ is constant.
\end{itemize}
We make a more detailed comparison between our and previous results in Table \ref{tab: main}.

 \begin{table}[ht]
	 \centering
	 \begin{tabular}{ccccc}
		 \hline
		 Constraint & Reference & Ratio & Complexity & Type \\ \hline\hline
		 Matroid & Buchbinder et al. \cite{BuchbinderFNS14} & $0.283-o(1)$ & $\mathcal{O}(nk)$ & Rand \\ 
		 Matroid & Buchbinder and Feldman \cite{BuchbinderF19} & $0.385$ & $\mathrm{poly}(n)$ & Rand \\ 
		 Matroid & Mirzasoleiman et al. \cite{MBK16} & $1/6-\epsilon$ & $\mathcal{O}(nk+k/\epsilon)$ & Det \\ 
		 Matroid & Lee et al. \cite{LMNS10} & $1/4-\epsilon$ & $\mathcal{O}((n^4 \log n)/\epsilon)$ & Det \\ 
		 Matroid & Han et al. \cite{HCCW20} & $1/4$ & $\mathcal{O}(nk)$ & Det \\ 
		 Matroid & Han et al. \cite{HCCW20} & $1/4-\epsilon$ & $\mathcal{O}((n/\epsilon)\log(k/\epsilon))$ & Det \\ \hline
		 Matroid & Theorem \ref{thm: matroid} & $0.283-o(1)$ & $\mathcal{O}(n^2k^2)$ & Det \\
		 \hline\hline
		 Knapsack & Amanatidis et al. \cite{AFLLR20} & $0.171$ & $\mathcal{O}(n\log n)$ & Rand \\
		 Knapsack & Buchbinder and Feldman \cite{BuchbinderF19} & $0.385$ & $\mathrm{poly}(n)$ & Rand \\
		 Knapsack & Amanatidis et al. \cite{AmanatidisKS19} & $1/7$ & $\mathcal{O}(n^4)$ & Det \\
		 Knapsack & Gupta et al. \cite{GuptaRST10} & $1/6$ & $\mathcal{O}(n^5)$ & Det \\
		 \hline
		 Knapsack & Theorem \ref{thm: enumeration} & $1/4$ & $\mathcal{O}(n^4)$ & Det \\ 
		 Knapsack & Theorem \ref{thm: enumeration} & $1/4-\epsilon$ & $\mathcal{O}((n^3/\epsilon)\log(n/\epsilon))$ & Det \\ 
		 \hline
		 \hline
		 Packing & Buchbinder and Feldman \cite{BuchbinderF19} & $0.385$ & $\mathrm{poly}(n)$ & Rand \\
		 \hline
		 Packing & Theorem \ref{thm: Constant Repeat Greedy for Linear Packing} & $1/6-\epsilon$ & $\mathcal{O}(n^2)$ & Det \\ 
		 \hline
		 \end{tabular}
	 \caption{Approximation algorithms for non-monotone submodular maximization under a matroid, a knapsack and linear packing constraints. The packing constraints have either a constant $m$ or a large width. ``Complexity'' refers to query complexity.
		 ``Rand'' is short for ``Randomized'' and ``Det'' is short for ``Deterministic''.}
	 \label{tab: main}
	 \end{table}

\subsection{Related Work}

To better illustrate the improvement of our results, this subsection provides a list of results in the literature on \emph{non-monotone} submodular maximization under a matroid constraint, a knapsack constraint, and linear packing constraints.

For the matroid constraint, the best randomized algorithm is based on the continuous greedy technique and achieves a $0.385$ approximation ratio \cite{BuchbinderF19}.
However, this algorithm suffers from high query complexity.
There exist different randomized algorithms called Random Greedy that achieve $1/4$ and $0.283-o(1)$ approximation ratios using $\mathcal{O}(nk)$ queries \cite{BuchbinderFNS14}, where $n$ is the total number of the elements and $k$ is the rank of the matroid.
When randomness is not allowed, Mirzasoleiman et al.~\cite{MBK16} proposed a deterministic algorithm that achieves a $1/6 -\epsilon$ approximation ratio and uses $\mathcal{O}(nk+k/\epsilon)$.
Lee et al.~\cite{LMNS10} proposed a deterministic $1/4-\epsilon$ approximation algorithm via the local search technique.
The algorithm uses $\mathcal{O}((n^4 \log n)/\epsilon)$ queries.
The best deterministic algorithm in the literature is called Twin Greedy \cite{HCCW20}, which achieves a $1/4$ approximation ratio using $\mathcal{O}(nk)$ queries and a $1/4-\epsilon$ ratio using $\mathcal{O}((n/\epsilon)\log(k/\epsilon))$ queries.

For the knapsack constraint, the same randomized $0.385$ approximation algorithm \cite{BuchbinderF19} also works by using a large number of queries.
A different randomized algorithm \cite{AFLLR20} uses nearly linear queries at a cost of a $0.171$ approximation ratio.
The best deterministic algorithm for the problem achieves a $1/6$ approximation ratio and uses $\mathcal{O}(n^5)$ queries \cite{GuptaRST10}.

For the linear packing constraints, no constant approximation is possible when the number of constraints $m$ is an input.
When $m$ is constant or the width of the constraints is large, constant approximation is possible \cite{LMNS10,SMM11,BuchbinderF19}.
Currently, the best known algorithm is again based on the continuous greedy algorithm and has an approximation ratio of $0.385$ \cite{BuchbinderF19}.
To the best of our knowledge, there is no deterministic algorithm for the problem in the literature.

\subsection{Organization}

In Section \ref{sec: pre}, we formally introduce the problem of non-monotone submodular maximization under a matroid constraint and a knapsack constraint. 
In Section \ref{sec: matroid}, we propose a deterministic $0.283-o(1)$ approximation algorithm for the matroid constraint. 
In Section \ref{sec: knapsack}, we propose deterministic algorithms for the knapsack constraint with different approximation ratios and query complexity.
In Section \ref{sec: Linear Packing}, we propose a deterministic $1/6-\epsilon$ approximation algorithm for the linear packing constraints when the width of the constraints is large.
In Section \ref{sec: conclusion}, we conclude the paper and list some future directions.

\section{Preliminaries}
\label{sec: pre}

In this section, we state the problems studied in this paper.

\begin{definition}[Submodular Function]
	Given a finite ground set $N$ of $n$ elements, a set function $f:2^N\mapsto \mathbb{R}$ is submodular if for all $S,T\subseteq N$,
	\[ f(S)+f(T) \geq f(S\cup T)+ f(S\cap T). \]
	Equivalently, $f$ is submodular if for all $S\subseteq T\subseteq N$ and $u\in N\setminus T$,
	\[ f(S\cup\{u\})-f(S)\geq f(T\cup\{u\})-f(T). \]
\end{definition}
For convenience, we use $f(S+u)$ to denote $f(S\cup\{u\})$, $f(u\mid T)$ to denote the marginal value $f(T+u)-f(T)$ of $u$ with respect to $T$, and $f(S\mid T)$ to denote the marginal value $f(S\cup T)-f(T)$.
The function $f$ is \emph{non-negative} if $f(S)\geq 0$ for all $S\subseteq N$.
$f$ is \emph{monotone} if $f(S)\leq f(T)$ for all $S\subseteq T\subseteq N$.


\begin{definition}[Matroid]
	A matroid system $\mathcal{M}=(N,\mathcal{I})$ consists of a finite ground set $N$ and a collection $\mathcal{I}\subseteq 2^N$ of the subsets of $N$, which satisfies the following three properties.
	\begin{itemize}
		\item $\emptyset \in \mathcal{I}$.
		\item If $A\subseteq B$ and $B\in \mathcal{I}$, then $A\in \mathcal{I}$.
		\item If $A,B\in\mathcal{I}$ and $|A|<|B|$, then there exists an element $u\in B\setminus A$ such that $A \cup \{u\} \in \mathcal{I}$.
	\end{itemize}
\end{definition}
For a matroid system $\mathcal{M}=(N,\mathcal{I})$, each $A\in \mathcal{I}$ is called an \emph{independent set}.
If $A$ is additionally maximal inclusion-wise, it is called a \emph{base}.
All bases of a matroid have an equal size, which is known as the \emph{rank} of the matroid.
In this paper, we use $k$ to denote the rank of a matroid.

\begin{definition}[Knapsack]
	Given a finite ground set $N$, assume there is a budget $B$, and each element $u\in N$ is associated with a cost $c(u)>0$.
	For set $S\subseteq N$, its cost $c(S)=\sum_{u\in S}c(u)$.
	We say $S$ is feasible if $c(S)\leq B$.
	The knapsack can be written as $\mathcal{I}=\{S\mid c(S)\leq B\}$.
\end{definition}
If $c(u)=1$ for all $u\in N$ in the knapsack and let $B=k$, the knapsack reduces to $\mathcal{I}=\{S\mid |S|\leq k\}$.
This is called the \emph{cardinality constraint} or the \emph{uniform matroid}, as it satisfies the definition of a matroid.

\begin{definition}[Linear Packing Constraints]
	Given a finite ground set $N$, a matrix $A\in [0, 1]^{m\times n}$, and a vector $b\in [1,\infty)^m$. For set $S \subseteq N $, the linear packing constraints can be written as $\mathcal{I}=\{S\mid Ax_S\leq b\}$, where $x_S$ stands for the characteristic vector of the set $S$.
	Let $W=\min \{b_i/A_{ij}\mid A_{ij}>0\}$.
	It is known as the \emph{width} of the packing constraints.
\end{definition}
When $m=1$, the linear packing constraints reduce to the knapsack constraint.

\begin{definition}[Constrained Submodular Maximization]
	The constrained submodular maximization problem has the form
	\[ \max\{f(S)\mid S\in\mathcal{I}\}. \]
\end{definition}
In this paper, the constraint $\mathcal{I}$ is assumed to be a matroid constraint, a single knapsack constraint, or linear packing constraints respectively. The objective function $f$ is assumed to be non-negative, non-monotone, and submodular.
Besides, $f$ is accessed by a value oracle that returns $f(S)$ when $S$ is queried.
The efficiency of any algorithm for the problem is measured by the number of queries it uses.

\section{Deterministic Approximation for Matroid Constraint}
\label{sec: matroid}

In this section, we present a deterministic $(0.283-o(1))$-approximation algorithm for submodular maximization under a matroid constraint.
Our algorithm is obtained by derandomizing the Random Greedy algorithm in \cite{BuchbinderFNS14}, using the technique from \cite{BuchbinderF18}.

To gain some intuition, we first describe the Random Greedy algorithm in \cite{BuchbinderFNS14}.
For convenience, we add a set $D$ of $2k$ ``dummy elements'' to the original instance $(N,f,\mathcal{M})$ to obtain a new instance $(N',f',\mathcal{M}')$ in the following way.
\begin{itemize}
	\item $N'=N\cup D$.
	\item $f'(S)=f(S\setminus D)$ for every set $S\subseteq N'$.
	\item $S\in\mathcal{I}'$ if and only if $S\setminus D\in\mathcal{I}$ and $|S|\leq k$.
\end{itemize}
Clearly, the new instance and the old one refer to the same problem.
Thus, in the remaining part of this section, we assume that any instance always contains the dummy elements as defined above.
The existence of dummy elements allows us to assume that the optimal solution is a base of $\mathcal{M}$ by adding dummy elements to it if it is not initially.
Another ingredient of the Random Greedy algorithm is the well-known exchange property of matroids, which is stated in Lemma \ref{lem: l1}.
\begin{lemma}[\cite{Schrijver03}]
	\label{lem: l1}
	If $A$ and $B$ are two bases of a matroid $\mathcal{M}=(N,\mathcal{I})$, then there exists a one-to-one function $g:A\rightarrow B$ such that
	\begin{itemize}
		\item $g(u)=u$ for every $u\in A\cap B$.
		\item for every $u\in A$, $B\cup\{u\}\setminus\{g(u)\}\in\mathcal{I}$.
	\end{itemize}
\end{lemma}

The Random Greedy algorithm is depicted as Algortihm \ref{alg: random greedy for matroid}. It was shown in \cite{BuchbinderFNS14} to be a $(0.283-o(1))$-approximation algorithm.

\begin{algorithm}[ht]
	\begin{algorithmic}[1]
		\State \textbf{Input} $N,f,\mathcal{M}$.
		\State Initialize $S_0$ to be an arbitrary base containing only dummy elements in $D$.
		
		\For{$i=1$ to $k$}
		\State Let $M_i\subseteq N\setminus S_{i-1}$ be a base of $\mathcal{M}$ containing only elements of $N\setminus S_{i-1}$ and maximizing $\sum_{u\in M_i} f(u\mid S_{i-1})$.
		\State	Let $g_i$ be a function mapping each element of $M_i$ to an element of $S_{i-1}$ such that $S_{i-1}+u-g_i(u)\in\mathcal{I}$ for every $u\in M_i$.
		\State	Let $u_i$ be a uniformly random element from $M_i$.
		\State $S_i\gets S_{i-1}+u_i-g_i(u_i)$.
		\EndFor
		\State \textbf{return} $S_k$.
		\caption{Random Greedy for Matroid \cite{BuchbinderFNS14}}
		\label{alg: random greedy for matroid}
	\end{algorithmic}
\end{algorithm}

We next explain how to use the technique from \cite{BuchbinderF18} to derandomize Algorithm \ref{alg: random greedy for matroid}.
The resulting algorithm is presented as Algorithm \ref{alg: deterministic matroid}.
It mimics Algorithm \ref{alg: random greedy for matroid} by explicitly maintaining a distribution $\mathcal{D}_i$ in each iteration $i$, which is supported on the subsets that the algorithm will visit.
The key point of Algorithm \ref{alg: deterministic matroid} is that it ensures that the size of the support of $\mathcal{D}_i$ grows polynomially with $i$ instead of exponentially as in Algorithm \ref{alg: random greedy for matroid}.

Formally, we use a multiset of pairs $\{(p_j,S_j)\}$ to represent distribution $\mathcal{D}_i$, where $S_j$ is a subset and $p_j$ is a probability.
In this representation, we allow $p_j=p_{j'}$, $S_j=S_{j'}$ or both for $j\neq j'$.
The sum of $p_j$ equals to $1$.
The support of $\mathcal{D}_i$ consists of distinct $S_j$ and is denoted by $\supp(\mathcal{D}_i)$.
The number of pairs in $\mathcal{D}_i$ is denoted by $|\mathcal{D}_i|$.
Clearly, $|\supp(\mathcal{D}_i)|\leq |\mathcal{D}_i|$.

Algorithm \ref{alg: deterministic matroid} updates the distribution and constructs $\mathcal{D}_i$ from $\mathcal{D}_{i-1}$ by finding an extreme-point solution of the linear program \eqref{lp}.
The variables of \eqref{lp} are $x(u,S)$ for $u\in M_i$ and $S\in\supp(\mathcal{D}_{i-1})$, which can be interpreted as the probability of adding $u$ into $S$.
The constraints of \eqref{lp} can be divided into six groups.
The last two groups ensure that $x(\cdot,S)$ is a legal probability vector for each $S\in\supp(\mathcal{D}_{i-1})$.
The first constraint guarantees that the addition of those $u$ to $S$ with positive $x(u,S)$ attains an average marginal value of elements in $M_i$.
The second constraint guarantees that the removal of $g_{i,S}(u)$ from $S$ causes a loss of at most an average marginal value.
The third group of constraints says that any element outside $S$ will be added into $S$ with probability at most $1/k$ and the fourth group of constraints says that any element in $S$ will be kicked out with probability at least $1/k$.
These constraints characterize the requirements for the elements to be selected in an iteration.
It is easy to see that $x(u,S)=1/k$ for every $S\in\supp(\mathcal{D}_{i-1})$ and $u\in M_i$ is a feasible solution of \eqref{lp} and this is exactly what the Random Greedy algorithm does.
However, this causes the size of the support of $\mathcal{D}_i$ to grow exponentially with $i$.
By choosing extreme-point solutions of \eqref{lp}, Algorithm \ref{alg: deterministic matroid} can ensure that it grows polynomially rather than exponentially.

There are two details about Algorithm \ref{alg: deterministic matroid} to which we need to pay attention.
First, for $A=M_i$ and $B=S$, let $g_{i,S}$ be the mapping defined in Lemma \ref{lem: l1}.
In line \ref{line: mapping}, Algorithm \ref{alg: deterministic matroid} needs to construct such mappings explicitly.
By invoking a perfect matching search algorithm for bipartite matching, they can be found in polynomial time.
Second, in line \ref{line: new dist}, if $u\in S$, $S+u-g_{i,S}(u)$ reduces to $S-g_{i,S}(u)$.

\begin{algorithm}[ht]
	\begin{algorithmic}[1]
		\State \textbf{Input} $N,f,\mathcal{M}$.
		\State Initialize $\mathcal{D}_0=\{(1,S)\}$, where $S$ is an arbitrary base containing only dummy elements in $D$.
		\For{$i=1$ to $k$}
		\State	Let $M_i\subseteq N$ be a base of $\mathcal{M}$ maximizing $\sum_{u\in M_i} \E_{S\sim\mathcal{D}_{i-1}}[f(u\mid S)]$.
		\State	Construct mapping $g_{i,S}$ for every $S\in\supp(\mathcal{D}_{i-1})$, where $g_{i,S}$ denotes the mapping defined in Lemma \ref{lem: l1} by plugging $A=M_i$ and $B=S$.\label{line: mapping}
		\State	Find an \textbf{extreme point} solution of the following linear formulation:
		\begin{alignat*}{2}
				\sum_{u\in M_i} \E_{S\sim\mathcal{D}_{i-1}}[x(u,S)\cdot f(u\mid S)] &\geq \frac{1}{k}\sum_{u\in M_i} \E_{S\sim\mathcal{D}_{i-1}}[f(u\mid S)] & \tag{P} \label{lp} \\
				\sum_{u\in M_i} \E_{S\sim\mathcal{D}_{i-1}}[x(u,S)\cdot f(g_{i,S}(u)\mid S\setminus\{g_{i,S}(u)\})] &\leq \frac{1}{k} \sum_{u\in M_i} \E_{S\sim\mathcal{D}_{i-1}}[f(g_{i,S}(u)\mid S\setminus\{g_{i,S}(u)\})]  & \\
				\E_{S\sim\mathcal{D}_{i-1}}[x(u,S)\cdot\mathbf{1}[u\notin S]] &\leq \frac{1}{k} \Pr_{S\sim\mathcal{D}_{i-1}}[u\notin S],\quad\forall u\in M_i & \\
				\E_{S\sim\mathcal{D}_{i-1}}[x(g_{i,S}^{-1}(u),S)\cdot\mathbf{1}[u\in S]] &\geq \frac{1}{k} \Pr_{S\sim\mathcal{D}_{i-1}}[u\in S], \quad\forall u\in N & \\
				\sum_{u\in M_i} x(u,S)&=1, \quad\forall S\in\supp(\mathcal{D}_{i-1}) &  \\
				x(u,S)&\geq 0, \quad\forall u\in M_i, S\in\supp(\mathcal{D}_{i-1}) &
		\end{alignat*}
		\State Construct a new distribution:\label{line: new dist}
		\[ \mathcal{D}_i\gets \{(x(u,S)\cdot \Pr_{\mathcal{D}_{i-1}}[S],S+u-g_{i,S}(u))\mid u\in M_i, S\in\supp(\mathcal{D}_{i-1}),x(u,S)>0\}. \]
		\EndFor
		\State \textbf{return} $\argmax_{S\in\supp(\mathcal{D}_k)} f(S)$.
		\caption{Derandomization of Random Greedy for Matroid}
		\label{alg: deterministic matroid}
	\end{algorithmic}
\end{algorithm}

We now introduce several useful properties that Algorithm \ref{alg: deterministic matroid} can guarantee.

\begin{lemma}[\cite{BuchbinderFNS14}]
	\label{lem: l2}
	Assume that $f$ is submodular. For any subset $T$ and a distribution $\mathcal{D}$,
	\[ \E_{S\sim\mathcal{D}}[f(T\cup S)]\geq f(T)\cdot\min_{u\in N}\Pr_{S\sim\mathcal{D}}[u\notin S]. \]
\end{lemma}

\begin{lemma}
	For every iteration $i=1,2,\ldots,k$ of Algorithm \ref{alg: deterministic matroid}, the following properties hold: 
	\begin{enumerate}
		\item The assignment $y(u,S)=1/k$ for every $u\in M_i$ and $S\in\supp(\mathcal{D}_{i-1})$ is a feasible solution of \eqref{lp}.
		\item The sum of the probabilities in $\mathcal{D}_i$ equals to $1$, and therefore $\mathcal{D}_i$ is a valid distribution.
		\item $|\mathcal{D}_i|\leq n+3k+2+|\mathcal{D}_{i-1}|$. Thus, $|\mathcal{D}_i|\leq ni+3ki+2i+1=\mathcal{O}(ni)$.
	\end{enumerate}
\end{lemma}

\begin{proof}
	We prove the lemma by induction on $i$.
	Assume that it holds for every $1\leq i'<i$.
	First, the correctness of the first property can be verified by directly plugging the assignment into \eqref{lp}.
	Next, notice that the sum of the probabilities in $\mathcal{D}_i$ is
	\[ \sum_{S\in\supp(\mathcal{D}_i)}\Pr_{\mathcal{D}_{i-1}}[S]\sum_{u\in M_i} x(u,S)=\sum_{S\in\supp(\mathcal{D}_i)}\Pr_{\mathcal{D}_{i-1}}[S]=1. \]
	Thus, the second property holds.
	Finally, recall that for an extreme point solution $x$, the number of variables $x(u,S)$ that are strictly greater than zero is no more than the number of (tight) constraints.
	Besides, observe that the number of constraints in \eqref{lp} at iteration $i$ is at most $n+3k+2+|\supp(\mathcal{D}_{i-1})|\leq n+3k+2+|\mathcal{D}_{i-1}|$.
	Since $S+u-g_{i,S}(u)$ is added to $\mathcal{D}_i$ only when $x(u,S)>0$, the size of $\mathcal{D}_i$ is also no more than $n+3k+2+|\mathcal{D}_{i-1}|$.
\end{proof}

\begin{lemma}
	\label{lem: prob of u not in S}
	For every element $u\in N\setminus D$ and $0\leq i\leq k$,
	\[ \Pr_{S\sim\mathcal{D}_i}[u\notin S]\geq \frac{1}{2}\left(1+\left(1-\frac{2}{k}\right)^i\right). \]
\end{lemma}

\begin{proof}
	For a fixed $u\in N\setminus D$,
	\begin{align*}
		\Pr_{S\sim\mathcal{D}_i}[u\in S] =\sum_{S\in\supp(\mathcal{D}_{i-1}):u\notin S} \Pr_{\mathcal{D}_{i-1}}[S] \cdot x(u, S) +\sum_{S\in\supp(\mathcal{D}_{i-1}):u\in S} \Pr_{\mathcal{D}_{i-1}}[S] \cdot \sum_{u'\in M_i-g_{i,S}^{-1}(u)}x(u',S).
	\end{align*}
	
	The first part equals to zero if $u\notin M_i$, and if $u\in M_i$, it holds that
	\begin{align*}
		\sum_{S\in\supp(\mathcal{D}_{i-1}):u\notin S} \Pr_{\mathcal{D}_{i-1}}[S] \cdot x(u, S) &=\sum_{S\in\supp(\mathcal{D}_{i-1})} \Pr_{\mathcal{D}_{i-1}}[S] \cdot x(u, S)\cdot\mathbf{1}[u\notin S] \\
		&=\E_{S\sim\mathcal{D}_{i-1}}[x(u,S)\cdot\mathbf{1}[u\notin S]] \\
		&\leq \frac{1}{k}\Pr_{S\sim\mathcal{D}_{i-1}}[u\notin S].
	\end{align*}
	The last inequality holds since $x$ is a feasible solution of \eqref{lp}.
	
	The second part satisfies
	\begin{align*}
		&\sum_{S\in\supp(\mathcal{D}_{i-1}):u\in S} \Pr_{\mathcal{D}_{i-1}}[S] \cdot \sum_{u'\in M_i-g_{i,S}^{-1}(u)}x(u',S) \\
		&=\sum_{S\in\supp(\mathcal{D}_{i-1}):u\in S} \Pr_{\mathcal{D}_{i-1}}[S] \cdot \left(1-x(g_{i,S}^{-1}(u),S)\right) \\
		&=\Pr_{S\sim\mathcal{D}_{i-1}}[u\in S]-\sum_{S\in\supp(\mathcal{D}_{i-1})}\Pr_{\mathcal{D}_{i-1}}[S]\cdot x(g_{i,S}^{-1}(u),S)\cdot\mathbf{1}[u\in S] \\
		&=\Pr_{S\sim\mathcal{D}_{i-1}}[u\in S]-\E_{S\sim\mathcal{D}_{i-1}}[x(g_{i,S}^{-1}(u),S)\cdot\mathbf{1}[u\in S]] \\
		&\leq (1-1/k)\Pr_{S\sim\mathcal{D}_{i-1}}[u\in S].
	\end{align*}
	The last inequality holds since $x$ is a feasible solution of \eqref{lp}.
	
	Let $p_{i,u}=\Pr_{S\sim\mathcal{D}_i}[u\in S]$.
	By the above argument,
	\[ p_{i,u}\leq (1-p_{i-1,u})/k+p_{i-1,u}(1-1/k)=p_{i-1,u}(1-2/k)+1/k. \]
	For $u\in N\setminus D$, $p_{0,u}=0$.
	It is easy to show by induction that
	\[ p_{i,u}\leq 0.5\cdot (1-(1-2/k)^i). \]
	The lemma follows immediately.
\end{proof}

Let $O$ be the optimal solution.
We have

\begin{lemma}
	For every iteration $i=1,2,\ldots, k$ of Algorithm \ref{alg: deterministic matroid},
	\[ \E_{S\sim\mathcal{D}_i}[f(S)]\geq \left(1-\frac{2}{k}\right)\cdot\E_{S\sim\mathcal{D}_{i-1}}[f(S)]+\frac{(1+(1-2/k)^{i-1})}{2k}\cdot f(O). \]
\end{lemma}

\begin{proof}
	On the one hand,
	\begin{align*}
		\sum_{u\in M_i} \E_{S\sim\mathcal{D}_{i-1}}[x(u,S)\cdot f(u\mid S)]
		&\geq \frac{1}{k}\sum_{u\in M_i} \E_{S\sim\mathcal{D}_{i-1}}[f(u\mid S)] \\
		&\geq \frac{1}{k}\sum_{u\in O} \E_{S\sim\mathcal{D}_{i-1}}[f(u\mid S)] \\
		&\geq \frac{1}{k} \E_{S\sim\mathcal{D}_{i-1}}[f(O\cup S)-f(S)] \\
		&\geq \frac{1}{k} \E_{S\sim\mathcal{D}_{i-1}}[0.5\cdot(1+(1-2/k)^{i-1})\cdot f(O)-f(S)].
	\end{align*}
	The first inequality holds since $x$ is a feasible solution of \eqref{lp}.
	The second is due to the choice of $M_i$.
	The third is due to submodularity.
	The last follows from Lemmas \ref{lem: l2} and \ref{lem: prob of u not in S}.
	
	On the other hand,
	\begin{align*}
		&\sum_{u\in M_i} \E_{S\sim\mathcal{D}_{i-1}}[x(u,S)\cdot f(g_{i,S}(u)\mid S\setminus\{g_{i,S}(u)\})] \\
		&\leq \frac{1}{k} \sum_{u\in M_i} \E_{S\sim\mathcal{D}_{i-1}}[f(g_{i,S}(u)\mid S\setminus\{g_{i,S}(u)\})] \\
		&= \frac{1}{k} \E_{S\sim\mathcal{D}_{i-1}}\left[\sum_{u\in M_i}(f(S)-f(S\setminus\{g_{i,S}(u)\}))\right] \\
		&\leq \frac{1}{k} \E_{S\sim\mathcal{D}_{i-1}}[f(S)].
	\end{align*}
	The first inequality holds since $x$ is a feasible solution of \eqref{lp}.
	The last inequality is due to submodularity.
	
	Finally, by combining the above inequalities,
	\begin{align*}
		&\E_{S\sim\mathcal{D}_i}[f(S)] \\
		&=\sum_{u\in M_i} \E_{S\sim\mathcal{D}_{i-1}}[x(u,S)\cdot f(S+u-g_{i,S}(u))] \\
		&\geq\sum_{u\in M_i} \E_{S\sim\mathcal{D}_{i-1}}[x(u,S)\cdot (f(S+u)+f(S-g_{i,S}(u))-f(S))] \\
		&= \sum_{u\in M_i} \E_{S\sim\mathcal{D}_{i-1}}[x(u,S)\cdot(f(S)+f(u\mid S)-f(g_{i,S}(u)\mid S\setminus\{g_{i,S}(u)\}))] \\
		&=\E_{S\sim\mathcal{D}_{i-1}}[f(S)]+\sum_{u\in M_i} \E_{S\sim\mathcal{D}_{i-1}}[x(u,S)\cdot f(u\mid S)]-\sum_{u\in M_i} \E_{S\sim\mathcal{D}_{i-1}}[x(u,S)\cdot f(g_{i,S}(u)\mid S\setminus\{g_{i,S}(u)\})] \\
		&\geq \E_{S\sim\mathcal{D}_{i-1}}[f(S)]+\frac{1}{k} \E_{S\sim\mathcal{D}_{i-1}}[0.5\cdot(1+(1-2/k)^{i-1})\cdot f(O)-f(S)]-\frac{1}{k} \E_{S\sim\mathcal{D}_{i-1}}[f(S)] \\
		&=\left(1-\frac{2}{k}\right)\cdot\E_{S\sim\mathcal{D}_{i-1}}[f(S)]+\frac{(1+(1-2/k)^{i-1})}{2k}\cdot f(O).
	\end{align*}
\end{proof}

The following theorem provides a theoretical guarantee for Algorithm \ref{alg: deterministic matroid}.

\begin{theorem}
	\label{thm: matroid}
	Algorithm \ref{alg: deterministic matroid} achieves a $(1+e^{-2})/4-\mathcal{O}(1/k^2)$ approximation ratio using $\mathcal{O}(n^2k^2)$ queries.
\end{theorem}

\begin{proof}
	We prove by induction that
	\[ \E[f(S_i)]\geq \frac{1}{4}\left[1+\left(\frac{2(i+1)}{k}-1\right)\left(1-\frac{2}{k}\right)^{i-1}\right]\cdot f(O). \]
	Clearly, the claim holds for $i=0$.
	Assume that the claim holds for every $0\leq i'<i$.
	By the above lemma,
	\begin{align*}
		\E_{S\sim\mathcal{D}_i}[f(S)] &\geq \left(1-\frac{2}{k}\right)\cdot\E_{S\sim\mathcal{D}_{i-1}}[f(S)]+\frac{0.5(1+(1-2/k)^{i-1})}{k}\cdot f(O) \\
		&\geq \left(1-\frac{2}{k}\right)\cdot\frac{1}{4}\left[1+\left(\frac{2i}{k}-1\right)\left(1-\frac{2}{k}\right)^{i-2}\right]\cdot f(O)+\frac{1+(1-2/k)^{i-1}}{2k}\cdot f(O) \\
		&=\frac{1}{4}\left[1+\left(\frac{2(i+1)}{k}-1\right)\left(1-\frac{2}{k}\right)^{i-1}\right]\cdot f(O).
	\end{align*}
	Therefore,
	\begin{align*}
		\E_{S\sim\mathcal{D}_{k}}[f(S)] &\geq \frac{1}{4}\left[1+\left(\frac{2(k+1)}{k}-1\right)\left(1-\frac{2}{k}\right)^{k-1}\right]\cdot f(O) \\
		&\geq \frac{1+\frac{e^{-2}(1+2/k)(1-4/k)}{1-2/k}}{4}\cdot f(O) \\
		&= \left(\frac{1+e^{-2}}{4}-\mathcal{O}(1/k^2)\right)\cdot f(O).
	\end{align*}
	Finally, since $|\mathcal{D}_i|=\mathcal{O}(in)$, Algorithm \ref{alg: deterministic matroid} makes $\mathcal{O}(n^2i)$ queries at iteration $i$.
	Therefore, the total number of queries made during all the iterations is $\mathcal{O}(n^2k^2)$.
\end{proof}

\section{Deterministic Approximation for Knapsack Constraint}
\label{sec: knapsack}

In this section, we present deterministic approximation algorithms for submodular maximization under a knapsack constraint.
In Section \ref{sec: twin greedy}, we present the Twin Greedy algorithm, which is originally designed for the matroid constraint in \cite{HCCW20} and for a mechanism design version of submodular maximization in \cite{AmanatidisKS19}.
It returns a set with a $1/4$ approximation ratio and uses $\mathcal{O}(n^2)$ queries.
In Section \ref{sec: threshold twin greedy}, we combine the threshold technique with the Twin Greedy algorithm and obtain the so-called Threshold Twin Greedy algorithm, which returns a set with a $1/4-\epsilon$ approximation ratio and uses $\Tilde{\mathcal{O}}(n/\epsilon)$ queries.
However, the sets returned by these two algorithms may be infeasible.
Therefore, in Section \ref{sec: feasible solutions}, we introduce the enumeration technique to turn them into feasible solutions with the same approximation ratio.
In Section \ref{sec: tight example}, we present a tight example that shows that the Twin Greedy algorithm can not achieve an approximation ratio better than $1/4$ even under the cardinality constraint.
Hence, our analysis for the algorithm, as well as the analysis in \cite{HCCW20}, is tight.

\subsection{The Twin Greedy Algorithm}
\label{sec: twin greedy}

In this section, we present a deterministic $1/4$ approximation algorithm for the knapsack constraint.
The formal procedure is presented as Algorithm \ref{alg: Twin Greedy Knapsack}.
It maintains two \emph{disjoint} candidate solutions $S_1$ and $S_2$ throughout its execution. At each round, when there remain unpacked elements and at least one feasible candidate solution, it determines a pair $(k,u)$ such that $u$ has the largest density with respect to $S_k$ ($k\in\{1,2\}$).
Then, $u$ is added into $S_k$ when its marginal value to $S_k$ is positive, otherwise, the algorithm will terminate immediately.
Note that the algorithm may return an infeasible set, since each candidate may pack one more element that violates the knapsack constraint.
We will handle this issue in Section \ref{sec: feasible solutions} by a standard enumeration technique.

\begin{algorithm}[ht]
	\caption{Twin Greedy for Knapsack}
	\begin{algorithmic}[1]
		\State \textbf{Input} $N,f,c,B$.
		\State $S_1 \gets \emptyset, S_2 \gets \emptyset, J \gets \{1,2\}$.
		\Comment{$k\in J$ iff $c(S_k)< B$}
		\While{$N\neq\emptyset$ and $J\neq \emptyset$}
		\State $(k,u)\gets \argmax_{k'\in J, u'\in N} \frac{f(u'\mid S_{k'})}{c(u')}$.
		\If{$f(u\mid S_k)\leq 0$}\label{line: if}
		\State \textbf{break}
		\EndIf
		\State $S_k\gets S_k\cup\{u\}$.
		\If{$c(S_k) \geq  B$}
		\State $J\gets J\setminus\{k\}$.
		\EndIf
		\State $N\gets N\setminus\{u\}$.
		\EndWhile
		\State\Return $\argmax\{f(S_1),f(S_2)\}$.
	\end{algorithmic}
	\label{alg: Twin Greedy Knapsack}
\end{algorithm}

We begin to analyze Algorithm \ref{alg: Twin Greedy Knapsack} by first defining some notations.
Assume that Algorithm \ref{alg: Twin Greedy Knapsack} ran for $t$ rounds in total.
In each round, an element was selected.
Thus, there are $t$ elements selected.
Some of them were added to $S_1$, and the others were added to $S_2$.
For $i=1,2,\ldots, t$, let $u_i$ be the element selected by Algorithm \ref{alg: Twin Greedy Knapsack} at round $i$.
For $k=1,2$ and $i=1,2,\ldots, t$, let $S_k^i$ be the $k$-th candidate solution at the end of round $i$ and $S_k$ denote the $k$-th candidate solution at the end of round $t$.
Thus, $S_k^i=S_k \cap \{u_1,u_2,\ldots,u_i\}$.
Let $S^*$ be the set returned by Algorithm \ref{alg: Twin Greedy Knapsack} and $O$ be the optimal solution.

We next introduce two useful properties (Lemmas \ref{lem: aux-1} and \ref{lem: aux-2}) that Algorithm \ref{alg: Twin Greedy Knapsack} can guarantee.

\begin{lemma}
	\label{lem: aux-1}
	For any $k\in\{1,2\}$, we have
	\begin{itemize}
		\item If $c(S_k)<B$, then $f(O\setminus(S_1\cup S_2)\mid S_k)\leq 0$.
		\item If $c(S_k)\geq B$, let $m$ be the smallest index such that $S_k^m=S_k$, then $f(O\setminus(S_1^m\cup S_2^m)\mid S_k)\leq \sum_{j: u_j\in S_k \setminus O} f(u_j\mid S_k^{j-1})$.
	\end{itemize}
\end{lemma}
\begin{proof}
	\textbf{Case} $c(S_k)<B$.
	If $O\setminus(S_1\cup S_2)=\emptyset$, then the claim holds trivially.
	If $O\setminus(S_1\cup S_2)\neq \emptyset$, since $c(S_k)<B$ and elements in $O\setminus(S_1\cup S_2)$ were not added into $S_1$ or $S_2$, by line \ref{line: if} of Algorithm \ref{alg: Twin Greedy Knapsack}, we know that $f(u\mid S_k)\leq 0$ for any $u\in O\setminus(S_1\cup S_2)$.
	By submodularity, $f(O\setminus(S_1\cup S_2)\mid S_k)\leq 0$.
	
	\textbf{Case} $c(S_k)\geq B$.
	For any $u_j\in S_k\setminus O$ and $u\in O\setminus(S_1^m\cup S_2^m)$, since $u_j$ is the element added into $S_k$ at round $j$ and $u$ was available but not selected at that time, we have
	\begin{align*}
		\frac{f(u_j\mid S_k^{j-1})}{c(u_j)}\geq \frac{f(u\mid S_k^{j-1})}{c(u)}.
	\end{align*}
	This is equivalent to
	\begin{align*}
		c(u)\cdot f(u_j\mid S_k^{j-1}) \geq c(u_j)\cdot f(u\mid S_k^{j-1}).
	\end{align*}
	Fixing $u_j$ and summing over all $u\in O\setminus(S_1^m\cup S_2^m)$, we have
	\begin{align*}
		c(O\setminus(S_1^m\cup S_2^m))\cdot f(u_j\mid S_k^{j-1}) &\geq c(u_j)\sum_{u\in O\setminus(S_1^m\cup S_2^m)} f(u\mid S_k^{j-1}) \\
		&\geq c(u_j)\sum_{u\in O\setminus(S_1^m\cup S_2^m)} f(u\mid S_k) \\
		&\geq c(u_j)\cdot f(O\setminus(S_1^m\cup S_2^m)\mid S_k).
	\end{align*}
	The last two inequalities are due to submodularity.
	Next, summing over all $u_j\in S_k\setminus O$, we have
	\begin{align*}
		\sum_{j: u_j\in S_k \setminus O} f(u_j\mid S_k^{j-1}) \geq \frac{c(S_k\setminus O)}{c(O\setminus(S_1^m\cup S_2^m))} \cdot f(O\setminus(S_1^m\cup S_2^m)\mid S_k).
	\end{align*}
	Since $c(S_k) \geq B \geq c(O)$, we have $c(S_k\setminus O) \geq c(O\setminus S_k) = c(O\setminus S_k^m) \geq c(O\setminus(S_1^m\cup S_2^m))$.
	Thus,
	\begin{align*}
		\sum_{j: u_j\in S_k \setminus O} f(u_j\mid S_k^{j-1}) \geq f(O\setminus(S_1^m\cup S_2^m)\mid S_k).
	\end{align*}
\end{proof}


\begin{lemma}
	\label{lem: aux-2}
	For any $k\in\{1,2\}$, let $\ell=3-k$ be the index of the other candidate solution.
	For fixed $i\in\{1,2,\ldots,t\}$, if $c(S_{k}^{i-1})<B$, then for any subset $T\subseteq S_{\ell}^i$, $f(T\mid S_{k}^i)\leq \sum_{j:u_j\in T} f(u_j\mid S_{\ell}^{j-1})$. 
\end{lemma}

\begin{proof}
	For any $u_j\in T\subseteq S_{\ell}^i$, since $u_j$ was added into $S_{\ell}$ instead of $S_{k}$ and $c(S_{k}^{j-1})\leq c(S_{k}^{i-1})<B$, we have
	\begin{align*}
		\frac{f(u_j\mid S_{\ell}^{j-1})}{c(u_j)}\geq \frac{f(u_j\mid S_{k}^{j-1})}{c(u_j)}.
	\end{align*}
	Thus, $f(u_j\mid S_{\ell}^{j-1})\geq f(u_j\mid S_{k}^{j-1})$.
	Besides, $S_{k}^{j-1}\subseteq S_{k}^i$ since round $j-1$ is before round $i$.
	Then,
	\begin{align*}
		\sum_{j:u_j\in T} f(u_j\mid S_{\ell}^{j-1}) \geq \sum_{j:u_j\in T} f(u_j\mid S_{k}^{j-1}) \geq \sum_{j:u_j\in T} f(u_j\mid S_{k}^i) \geq f(T\mid S_{k}^i).
	\end{align*}
	The last two inequalities are due to submodularity.
\end{proof}

The following theorem provides a theoretical guarantee for Algorithm \ref{alg: Twin Greedy Knapsack}.

\begin{theorem}
	Algorithm \ref{alg: Twin Greedy Knapsack} achieves a $1/4$ approximation ratio (though the output may be infeasible) and uses $\mathcal{O}(n^2)$ queries.
	\label{thm: Twin Greedy Knapsack}
\end{theorem}

\begin{proof}
	By our notations, we need to show that $f(S^*)\geq \frac{1}{4} f(O)$.
	Since $f(S^*)=\max\{f(S_1),f(S_2)\}\geq (f(S_1)+f(S_2))/2$, it suffices to show that $f(S_1)+f(S_2)\geq\frac{1}{2} f(O)$.
	We prove this by case analysis, according to whether the candidate solutions are feasible.
	
	\textbf{Case 1:} Both candidate solutions are feasible at the end of Algorithm \ref{alg: Twin Greedy Knapsack}, i.e., $c(S_1)<B$ and $c(S_2)<B$.
	
	In this case, by Lemma \ref{lem: aux-1}, we have
	\begin{align*}
		f(O\setminus S_2 \mid S_1) &=f(O\setminus (S_1\cup S_2)\mid S_1)\leq 0. \\
		f(O\setminus S_1 \mid S_2) &=f(O\setminus (S_1\cup S_2)\mid S_2)\leq 0.
	\end{align*}
	Besides, by plugging $T=O\cap S_{\ell}\subseteq S_{\ell}$ into Lemma \ref{lem: aux-2}, we have
	\begin{align*}
		f(O\cap S_2\mid S_1) &\leq \sum_{j:u_j\in O\cap S_2} f(u_j\mid S_{2}^{j-1}). \\
		f(O\cap S_1\mid S_2) &\leq \sum_{j:u_j\in O\cap S_1} f(u_j\mid S_{1}^{j-1}).
	\end{align*}
	Therefore, by combining the above inequalities,
	\begin{align*}
		f(S_1)+f(S_2)
		&\geq \sum_{j:u_j\in O\cap S_1} f(u_j\mid S_{1}^{j-1})+\sum_{j:u_j\in O\cap S_2} f(u_j\mid S_{2}^{j-1}) \\
		&\geq f(O\setminus S_1 \mid S_2)+f(O\setminus S_2 \mid S_1)+f(O\cap S_2\mid S_1)+f(O\cap S_1\mid S_2) \\
		&\geq f(O\mid S_1)+f(O\mid S_2) \\
		&= f(O\cup S_1)-f(S_1)+f(O\cup S_2)-f(S_2) \\
		&\geq f(O)-(f(S_1)+f(S_2)).
	\end{align*}
	The last two inequalities hold due to submodularity and the fact that $S_1\cap S_2=\emptyset$.
	Thus, $f(S_1)+f(S_2)\geq\frac{1}{2}f(O)$.
	
	\textbf{Case 2:} At least one candidate solution is not feasible at the end of Algorithm \ref{alg: Twin Greedy Knapsack}, i.e., $c(S_1)\geq B$ or $c(S_2)\geq B$.
	
	Assume that for some $m\in\{1,2,\ldots,t\}$, $S_1$ first became infeasible at the end of round $m$, that is, $c(S_1^{m-1})<B$ but $c(S_1^m)\geq B$, and $c(S_2^m)<B$. 
	By Lemma \ref{lem: aux-1},
	\begin{align*}
		f(O\setminus S_1 \mid  S_2) =f(O\setminus(S_1\cup S_2)\mid S_2) \leq \max\left\{0,\sum_{j:u_j\in S_2\setminus O} f(u_j\mid S_2^{j-1})\right\}\leq \sum_{j:u_j\in S_2\setminus O} f(u_j\mid S_2^{j-1}).
	\end{align*} 
	Again by Lemma \ref{lem: aux-1} and the fact that $S_1=S_1^m$,
	\begin{align*}
		f(O\setminus S_2^m \mid  S_1)=f(O\setminus(S_1^m\cup S_2^m)\mid S_1)\leq \sum_{j:u_j\in S_1\setminus O} f(u_j\mid S_1^{j-1}).
	\end{align*}
	Next, by submodularity and plugging $k=2,\ell=1, T=O\cap S_1\subseteq S_1^m$ into Lemma \ref{lem: aux-2},
	\begin{align*}
		f(O\cap S_1 \mid  S_2)\leq f(O\cap S_1 \mid  S_2^m)\leq \sum_{j:u_j\in O \cap S_1} f(u_j\mid S_1^{j-1}).
	\end{align*}
	By plugging $k=1,\ell=2, T=O\cap S_2^m\subseteq S_2^m$ into Lemma \ref{lem: aux-2},
	\begin{align*}
		f(O\cap S_2^m \mid  S_1) =f(O\cap S_2^m \mid  S_1^m) \leq \sum_{j:u_j\in O \cap S_2^m} f(u_j\mid S_2^{j-1})\leq \sum_{j:u_j\in O \cap S_2} f(u_j\mid S_2^{j-1}).
	\end{align*}
	Therefore, by combining the above inequalities,
	\begin{align*}
		&f(S_1)+f(S_2) \\
		&= \sum_{j:u_j\in S_1\setminus O} f(u_j\mid S_1^{j-1})+\sum_{j:u_j\in O\cap S_1} f(u_j\mid S_{1}^{j-1})+\sum_{j:u_j\in S_2\setminus O} f(u_j\mid S_2^{j-1})+\sum_{j:u_j\in O\cap S_2} f(u_j\mid S_{2}^{j-1}) \\
		&\geq f(O\setminus S_2^m \mid S_1)+f(O\cap S_1\mid S_2)+f(O\setminus S_1 \mid S_2)+f(O\cap S_2^m\mid S_1) \\
		&\geq f(O\mid S_1)+f(O\mid S_2) \\
		&= f(O\cup S_1)-f(S_1)+f(O\cup S_2)-f(S_2) \\
		&\geq f(O)-(f(S_1)+f(S_2)).
	\end{align*}
	The last two inequalities hold due to submodularity and the fact that $S_1\cap S_2=\emptyset$.
	Thus, $f(S_1)+f(S_2)\geq\frac{1}{2}f(O)$.
	
	Finally, Algorithm \ref{alg: Twin Greedy Knapsack} runs at most $n$ rounds and makes $\mathcal{O}(n)$ queries at each round.
	Thus, it makes $\mathcal{O}(n^2)$ queries in total.
\end{proof}



\subsection{The Threshold Twin Greedy Algorithm}
\label{sec: threshold twin greedy}

In this section, we accelerate the Twin Greedy algorithm by applying the threshold technique \cite{BadanidiyuruV14}.
The modified algorithm is called the Threshold Twin Greedy algorithm and is formulated as Algorithm \ref{alg: Threshold Twin Greedy Knapsack}.


Algorithm \ref{alg: Threshold Twin Greedy Knapsack} maintains a set of thresholds for each element and stores them decreasingly in a priority queue.
At each round, Algorithm \ref{alg: Threshold Twin Greedy Knapsack} only considers the first element in the queue.
It first removes the element from the queue and then compares its marginal density with its threshold.
If the marginal density is at least $(1-\epsilon)$ of the threshold, the element will be added to the current solution.
Otherwise, the threshold will be updated to the marginal density and the element will be reinserted into the queue as long as it has not been reinserted into the queue for $\Omega(\log n)$ times.

\begin{algorithm}[ht]
	\caption{Threshold Twin Greedy for Knapsack}
	\begin{algorithmic}[1]
		\State \textbf{Input} $N,f,c,B, \epsilon$.
		\State $S_1 \gets \emptyset, S_2 \gets \emptyset, J \gets \{1,2\}$.
		\Comment{$k\in J$ iff $c(S_k)< B$}
		\State $\Delta(u)\gets f(u)$ for $u\in N$.
		\State Maintain a priority queue $Q$ where elements in $N$ are sorted in decreasing order by their keys, and the key of $u\in N$ is initialized to $\Delta(u)/c(u)$.
		\State $q(u)\gets 0$ for $u\in N$.
		\Comment{$q(u)$ records the number of times $u$ is reinserted into $Q$}
		\State $D\gets \emptyset$.
		\Comment{$D$ records the elements that has been reinserted into $Q$ $\frac{2\ln(n/\epsilon)}{\epsilon}$ times}
		\While{$N\neq\emptyset$, $J\neq \emptyset$ and $Q\neq \emptyset$}
		\State Remove the element $u$ from $Q$ with the maximum key.
		\State $k\gets \argmax_{k'\in J} f(u\mid S_{k'})$.
		\If{$f(u\mid S_k)\leq 0$} \label{line: if-threshold}
		\State \textbf{break}
		\EndIf
		\If{$f(u\mid S_k)\geq (1-\epsilon)\cdot\Delta(u)$}
		\State $S_k\gets S_k\cup\{u\}$.
		\State $N\gets N\setminus\{u\}$.
		\If{$c(S_k) \geq  B$}
		\State $J\gets J\setminus\{k\}$.
		\EndIf
		\Else
		\If{$q(u)\leq \frac{2\ln(n/\epsilon)}{\epsilon}$}
		\State $\Delta(u)\gets f(u \mid S_k)$.
		\State Reinsert $u$ into $Q$ with key $\Delta(u)/c(u)$.
		\State $q(u)\gets q(u)+1$.
		\Else
		\State $D\gets D\cup\{u\}$.
		\EndIf
		\EndIf
		\EndWhile
		\State\Return $\argmax\{f(S_1),f(S_2)\}$.
	\end{algorithmic}
	\label{alg: Threshold Twin Greedy Knapsack}
\end{algorithm}

In this way, if the element was added into the solution, it means that Algorithm \ref{alg: Threshold Twin Greedy Knapsack} selects an element with the largest density at this round, up to a $(1-\epsilon)$ factor.
Otherwise, the threshold of this element will decrease by a factor of at least $1-\epsilon$.
In other words, the threshold decreases exponentially and it takes $\mathcal{O}(\log n)$ queries to determine whether Algorithm \ref{alg: Threshold Twin Greedy Knapsack} should select an element.

We begin to analyze Algorithm \ref{alg: Threshold Twin Greedy Knapsack} by first defining some notations.
Assume that Algorithm \ref{alg: Threshold Twin Greedy Knapsack} added $t$ elements into $S_1$ and $S_2$ in total, denoted by $\{u_1,u_2,\ldots,u_t\}$.
Some of them were added to $S_1$, and the others were added to $S_2$.
For $i=1,2,\ldots,t-1$, $u_i$ was added right before $u_{i+1}$.
But $u_i$ and $u_{i+1}$ are not necessarily added in two successive rounds.
For $i=1,2,\ldots, t$, let $\Delta_i(u)$ be the value of $\Delta(u)$ at the beginning of the round where $u_i$ is added.
For $k=1,2$ and $i=1,2,\ldots, t$, let $S_k^i$ be the $k$-th candidate solution after $u_i$ is added and $S_k$ denote the $k$-th candidate solution after $u_t$ is added.
Thus, $S_k^i=S_k \cap \{u_1,u_2,\ldots,u_i\}$.
Let $S^*$ be the set returned by Algorithm \ref{alg: Twin Greedy Knapsack}, $O$ be the optimal solution, and $O'=O\setminus D$, where $D$ is defined in Algorithm \ref{alg: Threshold Twin Greedy Knapsack}.

We first show that the marginal value of $O\setminus O'$ with respect to $S_k$ is small.
\begin{lemma}
	\label{lem: D is small}
	For any $k\in\{1,2\}$, we have $f(O\setminus O'\mid S_k)\leq \epsilon f(O)$.
\end{lemma}
\begin{proof}
	Since $O\setminus O'\subseteq D$, this means each $u\in O\setminus O'$ has been reinserted into $Q$ for $\frac{2\ln(n/\epsilon)}{\epsilon}$ times.
	Since $\Delta(u)$ decreases by a factor of at least $1-\epsilon$ each time $u$ is  reinserted into $Q$, we have $f(u\mid S_k)\leq (1-\epsilon)^{\frac{2\ln(n/\epsilon)}{\epsilon}} f(u)\leq\frac{\epsilon}{n} f(u)\leq \frac{\epsilon}{n}f(O)$.
	By submodularity, $f(O\setminus O'\mid S_k)\leq \sum_{u\in O\setminus O'} f(u\mid S_k)\leq \epsilon f(O)$.
\end{proof}


Next, we introduce two useful properties (Lemmas \ref{lem: aux-3} and \ref{lem: aux-4}) that Algorithm \ref{alg: Threshold Twin Greedy Knapsack} can guarantee.
These two lemmas are parallel to Lemmas \ref{lem: aux-1} and \ref{lem: aux-2}, by replacing $O$ with $O'=O\setminus D$. 

\begin{lemma}
	\label{lem: aux-3}
	For any $k\in\{1,2\}$, we have
	\begin{itemize}
		\item If $c(S_k)<B$, then $f(O'\setminus(S_1\cup S_2)\mid S_k)\leq 0$.
		\item If $c(S_k)\geq B$, let $m$ be the smallest index such that $S_k^m=S_k$, then we have $(1-\epsilon)\cdot f(O'\setminus(S_1^m\cup S_2^m)\mid S_k)\leq \sum_{j: u_j\in S_k \setminus O'} f(u_j\mid S_k^{j-1})$.
	\end{itemize}
\end{lemma}
\begin{proof}
	\textbf{Case} $c(S_k)<B$.
	If $O'\setminus(S_1\cup S_2)=\emptyset$, then the claim holds trivially.
	If $O'\setminus(S_1\cup S_2)\neq \emptyset$,  since $O'\cap D=\emptyset$, it means that each $u\in O'\setminus(S_1\cup S_2)$ is still in the queue $Q$.
	Since $c(S_k)<B$, by line \ref{line: if-threshold} of Algorithm \ref{alg: Threshold Twin Greedy Knapsack}, this means that $f(u\mid S_k)\leq 0$ for each $u\in O\setminus(S_1\cup S_2)$.
	By submodularity, $f(O'\setminus(S_1\cup S_2)\mid S_k)\leq 0$.
	
	\textbf{Case} $c(S_k)\geq B$.
	For any $u_j\in S_k\setminus O'$ and $u\in O'\setminus(S_1^m\cup S_2^m)$, since $u_j$ is the $j$-th element added into $S_k$ and $u$ was in the queue $Q$ but not selected at that time, we have
	\begin{align*}
		\frac{f(u_j\mid S_k^{j-1})}{c(u_j)}\geq \frac{(1-\epsilon)\cdot\Delta_j(u_j)}{c(u_j)}\geq \frac{(1-\epsilon)\cdot\Delta_j(u)}{c(u)}\geq \frac{(1-\epsilon)\cdot f(u\mid S_k^{j-1})}{c(u)}.
	\end{align*}
	The second inequality holds since $u_j$ has the maximum key in $Q$.
	The third inequality holds by submodularity and the observation that there must exist a $k'<k$ such that $\Delta_j(u)=f(u\mid S_{k'}^{j-1})$.
	By a similar argument to Lemma \ref{lem: aux-1}, we have
	\begin{align*}
		\sum_{j: u_j\in S_k \setminus O'} f(u_j\mid S_k^{j-1}) \geq \frac{c(S_k\setminus O')}{c(O'\setminus(S_1^m\cup S_2^m))} \cdot (1-\epsilon)\cdot f(O'\setminus(S_1^m\cup S_2^m)\mid S_k).
	\end{align*}
	Since $c(S_k) \geq B \geq c(O')$, we have $c(S_k\setminus O') \geq c(O'\setminus S_k) = c(O'\setminus S_k^m) \geq c(O'\setminus(S_1^m\cup S_2^m))$.
	Thus,
	\begin{align*}
		\sum_{j: u_j\in S_k \setminus O'} f(u_j\mid S_k^{j-1}) \geq (1-\epsilon)\cdot f(O'\setminus(S_1^m\cup S_2^m)\mid S_k).
	\end{align*}
\end{proof}


\begin{lemma}
	\label{lem: aux-4}
	For any $k\in\{1,2\}$, let $\ell=3-k$ be the index of the other candidate solution.
	For fixed $i\in\{1,2,\ldots,t\}$, if $c(S_{k}^{i-1})<B$, then for any subset $T\subseteq S_{\ell}^i$, $f(T\mid S_{k}^i)\leq \sum_{j:u_j\in T} f(u_j\mid S_{\ell}^{j-1})$.
\end{lemma}

\begin{proof}
	This is a restatement of Lemma \ref{lem: aux-2}.
\end{proof}

The following theorem provides a theoretical guarantee for Algorithm \ref{alg: Threshold Twin Greedy Knapsack}.

\begin{theorem}
	Algorithm \ref{alg: Threshold Twin Greedy Knapsack} achieves a $1/4-\epsilon$ approximation ratio (though the output may be infeasible) and uses $\mathcal{O}((n/\epsilon)\log(n/\epsilon))$ queries.
	\label{thm: Threshold Twin Greedy Knapsack}
\end{theorem}

\begin{proof}
    We need to show that $f(S^*)\geq (\frac{1}{4}-\epsilon) f(O)$.
	Since $f(S^*)=\max\{f(S_1),f(S_2)\}\geq (f(S_1)+f(S_2))/2$, it suffices to show that $f(S_1)+f(S_2)\geq (\frac{1}{2}-2\epsilon) f(O)$.
	We prove this by case analysis, according to whether the candidate solutions are feasible.
	
	\textbf{Case 1:} Both candidate solutions are feasible at the end of Algorithm \ref{alg: Twin Greedy Knapsack}, i.e., $c(S_1)<B$ and $c(S_2)<B$.
	
	In this case, by Lemma \ref{lem: aux-3}, we have
	\begin{align*}
		f(O'\setminus S_2 \mid S_1) &=f(O'\setminus (S_1\cup S_2)\mid S_1)\leq 0. \\
		f(O'\setminus S_1 \mid S_2) &=f(O'\setminus (S_1\cup S_2)\mid S_2)\leq 0.
	\end{align*}
	Besides, by plugging $T=O'\cap S_{\ell}\subseteq S_{\ell}$ into Lemma \ref{lem: aux-4}, we have
	\begin{align*}
		f(O'\cap S_2\mid S_1) &\leq \sum_{j:u_j\in O'\cap S_2} f(u_j\mid S_{2}^{j-1}). \\
		f(O'\cap S_1\mid S_2) &\leq \sum_{j:u_j\in O'\cap S_1} f(u_j\mid S_{1}^{j-1}).
	\end{align*}
	Therefore, by combining the above inequalities,
	\begin{align*}
		f(S_1)+f(S_2)
		&\geq \sum_{j:u_j\in O'\cap S_1} f(u_j\mid S_{1}^{j-1})+\sum_{j:u_j\in O'\cap S_2} f(u_j\mid S_{2}^{j-1}) \\
		&\geq f(O'\setminus S_1 \mid S_2)+f(O'\cap S_2\mid S_1)+f(O'\setminus S_2 \mid S_1)+f(O'\cap S_1\mid S_2) \\
		&\geq f(O'\mid S_1)+f(O'\mid S_2).
	\end{align*}
	By Lemma \ref{lem: D is small},
	\begin{align*}
		f(S_1)+f(S_2)+2\epsilon\cdot f(O)
		&\geq f(O'\mid S_1)+f(O'\mid S_2)+f(O\setminus O'\mid S_1)+f(O\setminus O'\mid S_2) \\
		&\geq f(O\mid S_1)+f(O\mid S_2) \\
		&= f(O\cup S_1)-f(S_1)+f(O\cup S_2)-f(S_2) \\
		&\geq f(O)-(f(S_1)+f(S_2)).
	\end{align*}
	
	The last two inequalities hold due to submodularity and the fact that $S_1\cap S_2=\emptyset$.
	By rearranging the inequality, $f(S_1)+f(S_2)\geq(\frac{1}{2}-\epsilon)\cdot f(O)$.
	
	\textbf{Case 2:} At least one candidate solution is not feasible at the end of Algorithm \ref{alg: Twin Greedy Knapsack}, i.e., $c(S_1)\geq B$ or $c(S_2)\geq B$.
	
	Assume that for some $m\in\{1,2,\ldots, t\}$, $S_1$ first became infeasible when $u_m$ is added, that is, $c(S_1^{m-1})<B$ but $c(S_1^m)\geq B$, and $c(S_2^m)<B$.
	By Lemma \ref{lem: aux-3},
	\begin{align*}
		f(O'\setminus S_1 \mid  S_2) &=f(O'\setminus(S_1\cup S_2)\mid S_2) \\
        &\leq \max\left\{0,\sum_{j:u_j\in S_2\setminus O} f(u_j\mid S_2^{j-1})\right\}\leq \sum_{j:u_j\in S_2\setminus O'} f(u_j\mid S_2^{j-1}).
	\end{align*}
	Again by Lemma \ref{lem: aux-3} and the fact that $S_1=S_1^m$,
	\begin{align*}
		(1-\epsilon)\cdot f(O'\setminus S_2^m \mid  S_1) =(1-\epsilon)\cdot f(O'\setminus(S_1^m\cup S_2^m)\mid S_1) \leq \sum_{j:u_j\in S_1\setminus O'} f(u_j\mid S_1^{j-1}).
	\end{align*}
	Next, by submodularity and plugging $k=1,\ell=2, T=O'\cap S_1\subseteq S_1^m$ into Lemma \ref{lem: aux-4},
	\begin{align*}
		f(O'\cap S_1 \mid  S_2)\leq f(O'\cap S_1 \mid  S_2^m)\leq \sum_{j:u_j\in O' \cap S_1} f(u_j\mid S_1^{j-1}).
	\end{align*}
	By plugging $k=1,\ell=2,T=O'\cap S_2^m\subseteq S_2^m$ into Lemma \ref{lem: aux-4},
	\begin{align*}
		f(O'\cap S_2^m \mid  S_1) =f(O'\cap S_2^m \mid  S_1^m)\leq \sum_{j:u_j\in O' \cap S_2^m} f(u_j\mid S_2^{j-1})\leq \sum_{j:u_j\in O' \cap S_2} f(u_j\mid S_2^{j-1}).
	\end{align*}
	Therefore, by combining the above inequalities,
	\begin{align*}
		&f(S_1)+f(S_2) \\
		&= \sum_{j:u_j\in S_1\setminus O'} f(u_j\mid S_1^{j-1})+\sum_{j:u_j\in O'\cap S_1} f(u_j\mid S_{1}^{j-1})+\sum_{j:u_j\in S_2\setminus O'} f(u_j\mid S_2^{j-1})+\sum_{j:u_j\in O'\cap S_2} f(u_j\mid S_{2}^{j-1}) \\
		&\geq (1-\epsilon)\cdot f(O'\cap S_2^m\mid S_1)+f(O'\cap S_1\mid S_2)+f(O'\setminus S_1 \mid S_2)+f(O'\setminus S_2^m \mid S_1)\\
		&\geq f(O'\mid S_1)+f(O'\mid S_2)-\epsilon\cdot f(O'\cap S_2^m\mid S_1) \\
		&\geq f(O'\mid S_1)+f(O'\mid S_2)-\epsilon \sum_{j:u_j\in O' \cap S_2} f(u_j\mid S_2^{j-1}) \\
		&\geq f(O'\mid S_1)+f(O'\mid S_2)-\epsilon\cdot f(S_2).
	\end{align*}
	Next, by Lemma \ref{lem: D is small},
	\begin{align*}
		&(1+\epsilon)(f(S_1)+f(S_2))+2\epsilon\cdot f(O) \\
		&\geq f(O'\mid S_1)+f(O'\mid S_2)+f(O\setminus O' \mid S_1)+f(O\setminus O' \mid S_2) \\
		&\geq f(O\mid S_1)+f(O\mid S_2) \\
		&=f(O\cup S_1)-f(S_1)+f(O\cup S_2)-f(S_2) \\
		&\geq f(O)-(f(S_1)+f(S_2)).
	\end{align*}
	The last two inequalities hold due to submodularity and the fact that $S_1\cap S_2=\emptyset$.
	By rearranging the inequality, $f(S_1)+f(S_2)\geq \frac{1-2\epsilon}{2+\epsilon} f(O)$.
	
	Finally, since each element will be reinserted to the queue $\mathcal{O}(\log(n/\epsilon)/\epsilon)$ times and there are $n$ elements, the total number of queries made by Algorithm \ref{alg: Threshold Twin Greedy Knapsack} is $\mathcal{O}((n/\epsilon)\log(n/\epsilon))$.
\end{proof}



\subsection{Obtaining Approximation-Preserving Feasible Solutions}
\label{sec: feasible solutions}

This section is dedicated to turning the sets returned by Algorithm \ref{alg: Twin Greedy Knapsack} and Algorithm \ref{alg: Threshold Twin Greedy Knapsack} into feasible solutions with the same approximation ratios.
The formal procedure is presented as Algorithm \ref{alg: enumeration}.
To get some intuitions, recall that the solutions of Algorithm \ref{alg: Twin Greedy Knapsack} and Algorithm \ref{alg: Threshold Twin Greedy Knapsack} become infeasible only when the last added element violates the knapsack constraint.
Thus, by throwing away the last element, the solutions become feasible.
However, this may cause a huge loss in the approximation ratio if the element has a significant value.
By enumerating all feasible solutions of size at most two, we can find elements of large values.
As a result, the remaining elements, including the last one added by Algorithm \ref{alg: Twin Greedy Knapsack} or Algorithm \ref{alg: Threshold Twin Greedy Knapsack}, will not be too large.
Thus, we can safely throw them away.
The overall procedure incurs an additional $n^2$ factor in the query complexity due to the enumeration step.

\begin{algorithm}[ht]
	\caption{Enumeration (Threshold) Twin Greedy for Knapsack}
	\begin{algorithmic}[1]
		\State \textbf{Input} $N,f,c,B$.
		\ForAll{$E\subseteq N$ with $|E|\leq 2$ and $c(E)\leq B$}
		\State Let $D=\{u\in N\setminus E\mid f(u\mid E)>\frac{1}{2} f(E)\}$.
		\State $G_E=\mbox{(Threshold-)Twin-Greedy}(N\setminus (E\cup D),f(\cdot\mid E), c(\cdot),B-c(E))$.
		\State Let $R_E= G_E$ if $c(G_E)\leq B-c(E)$ and otherwise $R_E=G_E\setminus\{u_E\}$, where $u_E$ is the last element added into $G_E$.
		\State Let $S_E=E\cup R_E$.
		\EndFor
		\State \textbf{return} $\argmax_E\{f(S_E)\}$.
	\end{algorithmic}
	\label{alg: enumeration}
\end{algorithm}

Formally, we show the following theorem.

\begin{theorem}
	\label{thm: enumeration}
	Algorithm \ref{alg: enumeration} achieves a $1/4$ approximation ratio and uses $\mathcal{O}(n^4)$ queries if Twin Greedy is invoked.
	It achieves a $1/4-\epsilon$ approximation ratio and uses $\mathcal{O}((n^3/\epsilon)\log(n/\epsilon))$ queries if Threshold Twin Greedy is invoked.
\end{theorem}

\begin{proof}
	We first note that $c(S_E)\leq B$ for every $S_E$, since by the definition of $R_E$, we have $c(R_E)\leq B-c(E)$.
	
	Let $O$ be the optimal solution.
	Assume that $|O|>2$, since otherwise $O$ is a candidate of $E$ and will be found immediately.
	Next, order elements in $O$ in a greedy manner such that $o_1=\argmax_{o\in O} f(o)$, $o_2=\argmax_{o\in O\setminus\{o_1\}} f(o\mid o_1)$, etc.
	Then, $\{o_1,o_2\}$ is a candidate of $E$ and will be visited during the \textbf{for} loop.
	In the following, we consider the round where $E=\{o_1,o_2\}$ and show that the corresponding $S_E$ achieves the desired ratio.
	This already suffices since the algorithm returns the maximum $S_E$.
	
	We claim that $f(o\mid E)\leq f(E)/2$ for any $o\in O\setminus E$.
	This follows from $f(o\mid E)\leq f(o\mid o_1)\leq f(o_2\mid o_1)$, $f(o\mid E)\leq f(o)\leq f(o_1)$ and $f(o_1)+f(o_2\mid o_1)=f(E)$.
	The claim implies that $D\cap(O\setminus E)=\emptyset$ or equivalently $O\setminus E \subseteq N\setminus(E\cup D)$.
	Besides, $c(O\setminus E)\leq B-c(E)$.
	Thus, when the Twin Greedy algorithm is invoked, $f(G_E\mid E)\geq f(O\setminus E\mid E)/4$.
	
	By the definition of $R_E$,
	\begin{align*}
		f(R_E\mid E) &\geq f(G_E\mid E)-f(u_E\mid E) \\
		&\geq\frac{1}{4}f(O\setminus E\mid E)-\frac{1}{2}f(E) \\
		&=\frac{1}{4}f(O)-\frac{3}{4}f(E).
	\end{align*}
	The first inequality is due to the submodularity of $f(\cdot\mid E)$.
	The second inequality holds since $u_E\in N\setminus (E\cup D)$ and hence $f(u_E\mid E)\leq f(E)/2$.
	The equality holds since $E\subseteq O$.
	Finally,
	\begin{align*}
		f(S_E)=f(E)+f(R_E\mid E)\geq \frac{1}{4} f(O)+\frac{1}{4} f(E)\geq \frac{1}{4} f(O).
	\end{align*}
	
	When the Threshold Twin Greedy algorithm is invoked, the theorem follows from the same argument, which we omit for simplicity.
\end{proof}

\subsection{A Tight Example for Twin Greedy}
\label{sec: tight example}

In this section, we present a tight example in Theorem \ref{thm: tight example}, showing that the Twin Greedy algorithm can not reach an approximation ratio better than $1/4$ even under the cardinality constraint.
Note that for the cardinality constraint, Twin Greedy always outputs a feasible solution.
Thus, the enumeration technique is unnecessary in this case.

\begin{theorem}
	\label{thm: tight example}
	There is an instance of non-monotone submodular maximization under a cardinality constraint such that if we run the Twin Greedy algorithm on this instance, the returned solution has a value of at most $1/4+o(1)$ of the optimum.
\end{theorem}

\begin{proof}
	Given a finite ground set $N$ of $n$ elements, arbitrarily choose two different elements $u_1,u_2 \in N$.
	For any $S\subseteq N$, let $T=S\setminus\{u_1,u_2\}$.
	Define a set function $f:2^{N}\rightarrow\mathbb{R}_+$ as follows.
	\begin{align*}
		f(S)=
		\begin{cases}
			0, & u_1,u_2 \in S \\
			|T|, & u_1,u_2 \notin S \\
			1 + \epsilon + \frac{1}{2}|T|, & \mbox{otherwise}
		\end{cases}
	\end{align*}
	It is easy to verify that $f$ is non-negative, non-monotone, and submodular.
	
	Assume that the constraint parameter $k=|N|=n$ is an even number.
	Clearly, the set $N\setminus\{u_1,u_2\}$ is an optimal solution with value $f(N\setminus\{u_1,u_2\})=k-2$.
	On the other hand, the Twin Greedy algorithm will add $u_1$ into $S_1$ and $u_2$ into $S_2$ in the first two rounds.
	After that, the algorithm may reach such a state that half of the remaining elements $N\setminus\{u_1,u_2\}$ are added into $S_1$ and the other half is added into $S_2$.
	In this case, $f(S_1)=f(S_2)=1+\epsilon+\frac{1}{2}(\frac{k}{2}-1)=\frac{k}{4}+\frac{1}{2}+\epsilon$.
	Hence, the returned solution has a value of at most $1/4+o(1)$ of the optimum.
	
	Finally, the instance can be easily generalized to arbitrary $k\le n$ by adding ``dummy'' elements.
\end{proof}

\section{Deterministic Approximation for Linear Packing Constraints}
\label{sec: Linear Packing}

In this section, we present a deterministic algorithm for submodular maximization under linear packing constraints with a large width.
Our algorithm is obtained by combining the multiplicative-updates algorithm \cite{AzarG12} for the monotone case and the technique from \cite{GuptaRST10} for dealing with the lack of monotonicity.
The multiplicative-updates algorithm is presented in Section \ref{sec: multiplicative updates}.
The overall algorithm is presented in Section \ref{sec: main algorithm for linear packing}.

\subsection{The Multiplicative Updates Algorithm}
\label{sec: multiplicative updates}

The multiplicative updates algorithm is depicted as Algorithm \ref{alg: multiplicative updates}.
It takes a parameter $\lambda$ as input, which is set to $\lambda=e^{\epsilon W}$ in the analysis.
It maintains a weight $w_i$ for the $i$-th constraint, which is updated in a multiplicative way.
Intuitively, it can be regarded as running a greedy algorithm over a ``virtual'' knapsack constraint, where element $j$ has a dynamic cost $\sum_{i=1}^m A_{ij} w_{i}$ and the knapsack has a dynamic budget $\sum_{i=1}^m b_i w_i$.

\begin{algorithm}[H]
	\caption{Multiplicative Updates for Linear Packing Constraints 
	}
	\begin{algorithmic}[1]
		\State \textbf{Input} $N, f, A, b,\lambda$
		\State $S \gets \emptyset$.
		\For{$i= 1$ to $m$} $w_{i} = 1/b_i$.
		\EndFor
		\While {$\sum_{i=1}^m b_i w_i \le \lambda$
			and $N\setminus S \neq \emptyset$}
		\State $j \gets \arg \max_{ j\in N\setminus S} \frac{f(j\mid S)}{\sum_{i=1}^m A_{ij} w_{i}}$.
		\If{$f(j\mid S)\le 0$} \textbf{break}\label{line: break}
		\EndIf
		\State $S \gets S \cup \{j\}$.
		\For{$i= 1$ to $m$}  $w_{i} = w_{i} \lambda^{A_{ij}/b_i}$.
		\EndFor
		\EndWhile
		\If{$A x_{S}\leq b$} \Return $S$.
		\Else{ \Return $S\setminus\{j\}$, where $j$ is the last element added into $S$.}
		\EndIf
	\end{algorithmic}
	\label{alg: multiplicative updates}
\end{algorithm}

Theorem \ref{thm: multiplicative updates} provides a theoretical guarantee for the set $S$ returned by Algorithm \ref{alg: multiplicative updates}.

\begin{theorem}
	For any fixed $\epsilon > 0$, assume that $W\ge \max \{\ln m/\epsilon^2, 1/\epsilon \}$ and set $\lambda=e^{\epsilon W}$.
	Then, the set $S$ returned by Algorithm \ref{alg: multiplicative updates} is feasible and satisfies $f(S)\geq \frac{1}{2}(1-3\epsilon)\cdot f(S\cup C)$ for any set $C$ satisfying $Ax_C\le b$.
	Algorithm \ref{alg: multiplicative updates} uses $\mathcal{O}(n^2)$ queries.
	\label{thm: multiplicative updates}
\end{theorem}
\begin{proof}
	We begin the proof by defining some notations.
	Assume that the algorithm added $t+1$ elements in total during the \textbf{while} loop.
	For each $r=1,2,\ldots,t+1$, let $S_r$ be the current solution immediately after the $r$-th iteration and $j_r$ be the element selected in the $r$-th iteration.
	For each $i=1,2,\ldots, m$ and $r=1,2,\ldots,t+1$, let $w_{ir}$ be the value of $w_i$ immediately after the $r$-th iteration and $\beta_r=\sum_{i=1}^m b_i w_{ir}$.
	Let $S$ be the set returned by Algorithm \ref{alg: multiplicative updates}.
	
	We first consider the feasibility of $S$.
	If $Ax_{S_{t+1}}\leq b$, then $S=S_{t+1}$ is certainly feasible.
	If $Ax_{S_{t+1}}> b$, then $S=S_t$.
	We will show that $S_t$ is feasible.
	Let $t'$ be the first iteration such that $S_{t'}$ is infeasible.
	Namely, $\sum_{j\in S_{t'}} A_{ij}>b_i$ for some $i$.
	Then, by the value of $w_{it'}$,
	\[ b_iw_{it'}=b_iw_{i0}\prod_{j\in S_{t'}}\lambda^{A_{ij}/b_i}=\lambda^{\sum_{j\in S_{t'}}A_{ij}/b_i}>\lambda. \]
	This implies $\beta_{t'}>\lambda$ and the \textbf{while} loop will terminate after this round.
	It implies that $t'=t+1$ and therefore $S_t$ is feasible.
	
	Next, we show that $f(S)\geq \frac{1}{2}(1-\epsilon)\cdot f(S\cup C)$.
	First consider the case where the \textbf{while} loop terminates with $\sum_{i=1}^m b_i w_{i(t+1)} \leq \lambda$.
	By the above argument, $\sum_{i=1}^m b_i w_{i(t+1)} \leq \lambda$ implies that $S_{t+1}$ is feasible.
	Thus, $S=S_{t+1}$.
	If the \textbf{while} loop does not break in line \ref{line: break}, then $S=N$ and the theorem holds.
	If the \textbf{while} loop breaks in line \ref{line: break}, by submodularity, $f(C\mid S)\leq \sum_{u\in C\setminus S} f(u\mid S)\leq 0$.
	Then, $f(S)\geq f(S\cup C)$ and the theorem holds.
	
	Now we only need to consider the case where the \textbf{while} loop terminates with $\sum_{i=1}^m b_i w_{i(t+1)} > \lambda$.
	In this case, the \textbf{while} loop does not break in line \ref{line: break}, and $S_t$ is returned.
	We will show that $S_t$ satisfies the desired property.
	We proceed by bounding the value of $\beta_t$.
	For the lower bound, by the definition of $W$ and $\lambda=e^{\epsilon W}$, we have
	\begin{align*}
		\beta^t e^\epsilon =\sum_{i=1}^m b_i w_{it} \cdot \left(e^{\epsilon W}\right)^{1/W}\geq \sum_{i=1}^m b_i w_{i(t+1)}> e^{\epsilon W}.
	\end{align*}
	Then, $\beta^t>e^{\epsilon (W-1)}$.
	For the upper bound, notice that for any $r=1,2,\ldots,t$,
	\begin{align*}
		\beta_r=\sum_{i=1}^m b_i w_{ir}
		&= \sum_{i=1}^m b_i w_{i(r-1)}\cdot \left(e^{\epsilon W}\right)^{A_{ij_r}/b_i} \\
		&\leq \sum_{i=1}^m b_i w_{i(r-1)}\cdot \left(1+\frac{\epsilon W A_{ij_r}}{b_i} + \left(\frac{\epsilon W A_{ij_r}}{b_i}\right)^2\right) \\
		&\leq \sum_{i=1}^m b_i w_{i(r-1)} + (\epsilon W + \epsilon^2 W)\sum_{i=1}^m A_{ij_r} w_{i(r-1)} \\
		&= \beta_{r-1} + (\epsilon W + \epsilon^2 W) \sum_{i=1}^m A_{ij_r} w_{i(r-1)}. 
	\end{align*}
	The first inequality holds since $e^x\le 1+x+x^2$ for any $x\in [0,1]$.
	The second inequality holds since $WA_{ij_r}/b_i \leq 1$ by definition and hence $\left(\epsilon W A_{ij_r}/{b_i}\right)^2 \le \epsilon^2 W A_{ij_r}/{b_i}$.
	We continue with bounding the value of $\sum_{i=1}^m A_{ij_r} w_{i(r-1)}$.
	By the choice of $j_r$ and submodularity, for any $j\in C\setminus S_{r-1}$,
	\[ \frac{f(j_r\mid S_{r-1})}{\sum_{i=1}^m A_{ij_r} w_{i(r-1)}}\geq\frac{f(j\mid S_{r-1})}{\sum_{i=1}^m A_{ij} w_{i(r-1)}}\geq \frac{f(j\mid S_t)}{\sum_{i=1}^m A_{ij} w_{i(r-1)}}. \]
	Since $Ax_C\le b$, it follows that
	\begin{align*}
		\frac{f(j_r\mid S_{r-1})}{\sum_{i=1}^m A_{ij_r} w_{i(r-1)}} \geq\frac{\sum_{j\in C\setminus S_{r-1}} f(j\mid S_t)}{\sum_{j\in C\setminus S_{r-1}}\sum_{i=1}^m A_{ij} w_{i(r-1)}} \geq \frac{f(C\mid S_t)}{\sum_{i=1}^m b_i w_{i(r-1)}}=\frac{f(C\mid S_t)}{\beta_{r-1}}.
	\end{align*}
	Therefore,
	\[ \sum_{i=1}^m A_{ij_r} w_{i(r-1)}\leq \frac{f(j_r\mid S_{r-1})}{f(C\mid S_t)}\cdot\beta_{r-1}. \]
	Plugging it back into the recurrence of $\beta_r$, we obtain that
	\begin{align*}
		\beta_r \leq \beta_{r-1} \cdot\left(1+ \frac{(\epsilon W + \epsilon^2 W) f(j_r\mid S_{r-1})}{f(C \mid S_t)} \right)\leq \beta_{r-1}\cdot \exp{\left(\frac{(\epsilon W + \epsilon^2 W) f(j_r\mid S_{r-1})}{f(C \mid S_t)}\right)}.
	\end{align*}
	The last inequality is due to $1+x\le e^x$.
	By expanding the above recurrence and the fact that $\beta_0=m\leq e^{\epsilon ^2 W}$, we have
	\begin{align*}
		\beta_{t} 
		&\leq \beta_{0} \prod_{r=1}^t \exp{\left(\frac{(\epsilon W + \epsilon^2 W) f(j_r\mid S_{r-1})}{f(C \mid S^{t})}\right)} \\
		&\leq \exp {\left(\epsilon^2 W +(\epsilon W + \epsilon^2 W) \sum_{r=1}^t \frac{f(j_r\mid S_{r-1})}{f(C \mid S^{t})}\right)} \\
		&\leq \exp {\left(\epsilon^2 W +(\epsilon W + \epsilon^2 W) \frac{f(S_t)}{f(C \mid S^{t})}\right)}.
	\end{align*}
	Combining with $\beta^t\ge e^{\epsilon (W-1)}$, we obtain that
	\begin{align*}
		\frac{\epsilon(W-1)-\epsilon^2 W}{\epsilon W + \epsilon^2 W}
		\leq \frac{f(S^t)}{f(C \mid S^{t})}.
	\end{align*}
	Note that $(\epsilon(W-1)-\epsilon^2W)/(\epsilon W + \epsilon^2 W )\ge (1-2\epsilon)/(1+\epsilon)\ge 1-3\epsilon$. Thus, we get 
	$f(S^t)\ge (1-3\epsilon)\cdot f(C\mid S_t)
	=(1-3\epsilon)(f(S_t\cup C)-f(S_t)) 
	$ and hence $f(S_t)\geq \frac{1}{2}(1-3\epsilon)\cdot f(S_t\cup C)$.
	
	Finally, Algorithm \ref{alg: multiplicative updates} runs at most $n$ rounds and makes $\mathcal{O}(n)$ queries at each round.
	Thus, it makes $\mathcal{O}(n^2)$ queries in total.
\end{proof}

\subsection{The Main Algorithm for Linear Packing Constraints}
\label{sec: main algorithm for linear packing}

The multiplicative updates algorithm itself can not produce a solution with a constant approximation ratio, due to the lack of monotonicity.
To overcome this difficulty, we apply the technique from \cite{GuptaRST10}.
The resulting algorithm is depicted as Algorithm \ref{alg: main algorithm for linear packing}.
It first invokes the multiplicative updates algorithm to obtain a solution $S_1$.
Then it runs the multiplicative updates algorithm again over the remaining elements $N\setminus S_1$ to obtain another solution $S_2$.
These two solutions still can not guarantee any constant approximation ratios.
To remedy this, the algorithm produces the third solution $S_1'$ by solving the unconstrained submodular maximization problem over $S_1$.
This problem can be solved by an optimal deterministic $1/2$-approximation algorithm \cite{BuchbinderF18} using $O(n^2)$ queries.
Our algorithm finally returns the maximum solution among $S_1$, $S_2$, and $S'_1$.
We will show that this achieves a $1/6-\epsilon$ approximation ratio.

\begin{algorithm}[H]
	\caption{Main Algorithm for Linear Packing Constraints}
	\begin{algorithmic}[1]
		\State \textbf{Input} $N, f, A, b,\epsilon$.
		\State $S_1\gets \mbox{Multiplicative-Updates}(N, f, A, b,e^{\epsilon W/3})$.
		\State $S_2\gets \mbox{Multiplicative-Updates}(N\setminus S_1, f, A, b,e^{\epsilon W/3})$.
		\State $S_1'\gets \mbox{Unconstrained-Submodular-Maximization}(S_1, f)$.
		\State \textbf{return} $\max \{S_1, S_2, S_1'\}$.
	\end{algorithmic}
	\label{alg: main algorithm for linear packing}
\end{algorithm}

\begin{theorem}
	For any fixed $\epsilon > 0$, assume that $W\ge \max \{9\ln m/\epsilon^2, 3/\epsilon \}$.
	Algorithm \ref{alg: main algorithm for linear packing} achieves a $1/6-\epsilon$ approximation ratio and uses $\mathcal{O}(n^2)$ queries.
	\label{thm: Constant Repeat Greedy for Linear Packing}
\end{theorem}
\begin{proof}
	Let $O\in\arg\max\{f(S):Ax_S\leq b\}$.
	By Theorem \ref{thm: multiplicative updates}, we have
	\[ f(S_1)\geq \frac{1}{2}(1-\epsilon)\cdot f(S_1 \cup O) \mbox{ and } f(S_2)\ge \frac{1}{2}(1-\epsilon)\cdot f(S_2 \cup (O \setminus S_1)). \]
	If $f(S_1\cap O)\geq \delta\cdot f(O)$, then $f(S_1')\geq \frac{1}{2}\cdot f(S_1\cap O)\geq \frac{\delta}{2}\cdot f(O)$.
	If $f(S_1\cap O) \leq  \delta\cdot f(O)$, then
	\begin{align*}
		f(S_1)\ge \frac{1-\epsilon}{2} \cdot f(S_1 \cup O)
		\ge \frac{1-\epsilon}{2} \cdot (f(S_1 \cup O) + f(S_1\cap O) - \delta\cdot f(O)).
	\end{align*}
	We thus have that
	\begin{align*}
		\max\{f(S_1), f(S_2)\}
		&\ge \frac{1}{2} (f(S_1)+f(S_2)) \\
		&\ge \frac{1-\epsilon}{4} \cdot (f(S_1 \cup O ) + f(S_1\cap O) + f(S_2 \cup (O \setminus S_1))- \delta\cdot f(O)) \\
		&\ge \frac{1-\epsilon}{4} \cdot (f(S_1 \cup S_2 \cup O) + f(O\setminus S_1) + f(S_1\cap O) - \delta\cdot f(O)) \\
		&\ge \frac{1-\epsilon}{4} \cdot (f(S_1 \cup S_2 \cup O) + f(O) - \delta\cdot f(O)) \\
		&\ge \frac{(1-\epsilon)(1-\delta)}{4} \cdot f(O).
	\end{align*}
	The third and fourth inequalities hold due to submodularity.
	The last inequality holds due to non-negativity.
	Therefore, the approximation ratio of the returned set is at least $\max \{\delta/2, (1-\epsilon)(1-\delta)/4\}$.
	Let $\delta = (1-\epsilon)/(3-\epsilon)$.
	We get that the approximation ratio is at least $(1-\epsilon)/6$.
	
	Finally, since both the algorithm from \cite{BuchbinderF18} and Algorithm \ref{alg: multiplicative updates} make $\mathcal{O}(n^2)$ queries, Algorithm \ref{alg: Twin Greedy Knapsack} also makes $\mathcal{O}(n^2)$ queries in total.
\end{proof}

\section{Conclusion and Future Work}
\label{sec: conclusion}

In this paper, we propose deterministic algorithms with improved approximation ratios for non-monotone submodular maximization under a matroid constraint, a single knapsack constraint, and linear packing constraints, respectively.
We also show that the analysis of our knapsack algorithms is tight.

A central open question in this field is whether deterministic algorithms can achieve the same approximation ratio as randomized algorithms.
When the objective function is non-monotone, though our algorithms improve the best known deterministic algorithms, their approximation ratios are still worse than the best randomized algorithms.
When the function is monotone, the state-of-the-art deterministic algorithm for the matroid constraint achieves a $0.5008$ approximation ratio \cite{BuchbinderF019}, which is also smaller than the optimal $1-1/e$ ratio.
It is very interesting to fill these gaps.

%
%
%
\bibliographystyle{plain}
\bibliography{lib-reference}

\end{document}